\documentclass[11pt,letterpaper]{article}
\usepackage[utf8]{inputenc}
\usepackage[margin=1in]{geometry}
\usepackage{amsmath}
\usepackage{palatino}
\usepackage{macros}
\usepackage{latexsym} 
\usepackage{enumitem}
\usepackage{authblk}
\allowdisplaybreaks

\newcommand{\graph}{G}

\newcommand{\const}{c}
\newcommand{\vect}{f}
\newcommand{\dimension}{d}
\newcommand{\prtin}{[d]}
\newcommand{\ti}{i}
\newcommand{\tj}{j}
\def\vol{\textup{vol}}

\newcommand{\tl}{\ell}
\def\eps{\epsilon}
\newcommand{\ttt}{t}

\newcommand{\prt}{T}

\newcommand{\boundMatrix}{M}

\newcommand{\declareperson}[1]{\expandafter\newcommand\csname#1\endcsname[1]{\textcolor{orange}{#1: ##1}}}
\declareperson{Shayan}
\declareperson{Dorna}

\title{An Improved Trickle Down Theorem for Partite Complexes}
 \author{Dorna Abdolazimi \thanks{\href{mailto:dornaa@cs.washington.edu}{dornaa@cs.washington.edu}. Research supported by NSF grant CCF-1907845 and Air Force Office of Scientific Research grant FA9550-20-1-0212.}}
\author{Shayan Oveis Gharan \thanks{\href{mailto:shayan@cs.washington.edu}{shayan@cs.washington.edu}. Research supported by Air Force Office of Scientific Research grant FA9550-20-1-0212, Simons Investigator grant.}}
\affil{University of Washington}

\begin{document}
\maketitle

\begin{abstract}
We prove a strengthening of the trickle down theorem for partite complexes. Given a $(\dimension+1)$-partite $\dimension$-dimensional simplicial complex, we show that if  ``on average'' the links of  faces of co-dimension 2 are $\frac{1-\delta}{d}$-(one-sided) spectral expanders, then the link of any face of co-dimension $k$ is an  $O(\frac{1-\delta}{k\delta})$-(one-sided) spectral expander, for all $3\leq k\leq d+1$. For an application, using our theorem as a black-box, we show that links of faces of co-dimension $k$ in recent constructions of bounded degree high dimensional expanders have spectral expansion at most $O(1/k)$ fraction of the spectral expansion of the links of the worst faces of co-dimension $2$.
\end{abstract}

\section{Introduction}
A simplicial complex $X$ on a finite ground set $[n]=\{0,\dots,n\}$ is a downwards closed collection of subsets of $[n]$, i.e. if $\tau\in X$ and $\sigma \subset \tau$, then $\sigma \in X$. The elements of $X$ are called {\em faces}, and the maximal faces are called {\em facets}. We say that a face $\tau$ is of dimension $k$ if $|\tau| = k + 1$ and write $\dimfn(\tau) = k$. 
 A simplicial complex $X$ is a {\em pure} $\dimension$-dimensional complex if every facet  has dimension $\dimension$. In this paper, all  simplicial complexes are assumed to be pure. Given a  $\dimension$-dimensional complex $X$, for any   $0 \leq i \leq \dimension$,    define  $X(i) = \{\tau \in X : \dimfn(\tau) = i\}$.  Moreover, the co-dimension of a face $\tau \in X$ is defined as $\codim (\tau)  \coloneqq \dimension - \dimfn(\tau)$. For a face $\tau \in X$, define the  {\em link} of $\tau$ as the simplicial complex $X_{\tau} = \{\sigma \setminus \tau : \sigma \in X, \sigma \supset \tau\}$. Note that $X_\tau$ is a $(\codim(\tau) -1)$- dimensional complex.
 
  A  $(\dimension+1)$-partite  complex is a a  $\dimension$-dimensional complex such that $X(0)$ can be (uniquely) partitioned into sets $\prt_0 \cup \dots \cup \prt_{\dimension}$ such that for  every facet  $\tau \in X (\dimension)$,  we have $|\tau \cap \prt_i|=1$ for all $i \in [\dimension]$. The type of any face $\tau \in X$ is defined as $\type (\tau) = \{\ti \in \prtin : |\tau \cap \prt_\ti|=1\}$. 

We equip $X$ with a probability distribution $\pi$  supported on all {\em facets} of $X$ and we denote this pair by $(X,\pi)$. 
 For a face $\tau\in X$, $\pi$ induces a conditional distribution $\pi_{\tau}$ on facets of $X_\tau$ where for each facet $\sigma\in X_\tau$, 
 $$\pi_\tau(\sigma)=\frac{\pi(\sigma \cup \tau )}{\sum_{\text{facet }\sigma'\in X_\tau} \pi(\sigma' \cup \tau)}.$$
For each face $\tau$ of co-dimension at least 2 
the  1-skeleton of $(X_\tau, \pi_\tau)$ is a  weighted graph with vertices  $X_\tau (0)$, edges  $X_\tau (1)$, and edge weights given by $\P_{\sigma\sim\pi_\tau}[\{x,y\}\subseteq \sigma]$ for  each edge $\{x,y\}$. Note that when $\tau$ is of co-dimension 2,  the complex $(X_\tau, \pi_{\tau})$ is itself a weighted graph.  We say that a complex $X$ is {\it totally connected} if the 1-skeleton of the link of any face $\tau$ of co-dimension at least 2 is connected. 

 \begin{definition}[Local Spectral High Dimensional Expander]
 We say that the link of a face $\tau$ of co-dimension at least  $ 2$ of a  $\dimension$-dimensional (weighted) complex $(X,\pi)$ is a $\lambda$-(one sided) spectral expander if the second largest eigenvalue of the simple random walk on the 1-skeleton of $(X_\tau,\pi_\tau)$ is at most $\lambda$. We say that $(X,\pi)$
 is a $(\gamma_2,\gamma_3,\dots,\gamma_{d+1})$-local spectral expander if the link of  any face $\tau$ of co-dimension at least 2  is a $\gamma_{\codim(\tau)}$-spectral expander.
  When the complex $(X, \pi)$ is clear in the context, for an integer $2\leq k\leq \dimension+1$, we write $\gamma_k$ to denote the largest 2nd eigenvalue of the simple random walk on the 1-skeleton of all links of faces of co-dimension $k$ of the complex. 
 \end{definition}
 Over the last few years, the study of local spectral high dimensional expanders  (HDX) has revolutionized several areas of Math and theoretical computer science, namely in analysis of Markov chains \cite{ALOV19,  ALO20,   CLV21,ALO21}, coding theory \cite{DELLM22}, and elsewhere \cite{AJT19, DHKNT21, DFHT21}. 
 One can generally divide the family of HDXes studied in recent works into two groups: (i) Dense Complexes. Here, we have a HDX with exponentially large number of facets, i.e., $|X(0)|^d$. One typically encounters these objects in studying Markov Chain Monte Carlo technique where we use a Markov Chain to sample from an exponentially large probability distribution. Perhaps the simplest such family is the complex of all independent sets of a matroid.
 (ii) Sparse/Ramanujan Complexes. Here we have a HDX where every vertex (of $X(0)$) only appear in constant number of facets, independent of $|X(0)|$. See, \cite{LSV05,KO18,OP22} for explicit constructions. These  objects have been useful in constructing double samplers \cite{DHKNT21}, agreement testers  \cite{DK17, DD19}, or locally testable codes \cite{DELLM22}.
 
 One of the main aspects of local spectral expanders  is their ``local to global phenomenon'', often referred to as the Garland's method or the trickle down theorem \cite{Opp18}.
 \begin{theorem}[Trickle Down Theorem]\label{thm:trickledown}
 Given a totally connected complex $(X, \pi)$, if   $\gamma_2\leq \frac{1-\delta}{d}$ for some $0<\delta\leq 1$, then $\gamma_k \leq \frac{1-\delta}{d- (k-2)(1-\delta)} \leq \frac{1-\delta}{d\delta}$ for all $2 \leq k \leq \dimension$. 
 \end{theorem}

 The trickle down theorem has found numerous applications in proving bounds on local spectral expansion of simplicial complexes. To invoke the theorem one needs to inspect all faces of co-dimension 2 to find the worst 2nd eigenvalue. If we get lucky and this number is below $1/d$, then, the trickle down theorem kicks in and inductively bounds the spectral expansion of all links of the complex. 
 
There are, however, two pitfalls for the theorem: i) The required bound on $\gamma_2$ is too small and often not satisfiable. In particular, for many dense complexes in counting and sampling applications  that satisfy $\gamma_k=O(1/k)$ for $k\geq \Omega(d)$ (see e.g., \cite{ALO20,CLV21}), the links of faces of co-dimension 2 are only $\Theta(1)$-spectral expanders. ii) Even if $\gamma_2\ll 1/d$, the trickle down theorem only implies $\gamma_k \simeq \gamma_2$, i.e., $\gamma_k$ does not increase too much as $k$ increases.   This is in contrast with the fact that, for many dense complexes, one can observe a steep decrease in spectral expansion as the co-dimension increases, i.e., $\gamma_k \lesssim \gamma_2/k$. 

Such a decrease has not been known for any sparse complex. This led some experts to conjecture that, perhaps, dense and sparse complexes exhibit a different pattern of local spectral expansion; in particular, unlike dense HDX, local spectral expansion does not decay for sparse complexes.


In this paper, we prove a generalization of the trickle down theorem for {\em partite complexes} that shows that even if $\gamma_2=\Theta(1)$, but ``on average'' the links of  faces of co-dimension 2 are $<1/d$-spectral expanders, then we have  $\gamma_k\leq O(1/k)$ for all $3\leq k\leq d+1$. Surprising to us, our average condition is satisfied by some recent construction of (sparse) bounded degree high dimensional expanders \cite{KO18,OP22}. In particular, as we explain below, one can use our theorem to prove a significantly better local spectral expansion for the Kaufman-Opennheim construction in a black-box manner.

  \subsection{Main Contribution}
  
  We start by stating two special cases of our theorem. We need the following definition. 

\begin{definition}\label{def:Delta(i)}
    Given a $(\dimension+1)$-partite complex $(X, \pi)$ with parts $\prtin$, for every $i \in \prtin$, define $$\Delta_{(X, \pi)}(\ti) =  | \{ \tj \in \prtin \setminus \ti :  \exists \tau \text{ of } \type (\tau) = \prtin \setminus \{i, j\} \text{ s.t. } \lambda_2(P_\tau) >0  \}|,$$ i.e. $\Delta_{(X, \pi)}(\ti)$  is the number of parts $j \neq i$ for which there exists a face of type $\prtin \setminus \{i, j\}$ whose link is not a 0-spectral expander. Moreover, define $\Delta_{(X, \pi)} = \max_{\ti \in \prtin} \Delta_{(X, \pi)} (\ti) $. We drop the subscripts $(X, \pi)$ when the complex is clear in the context. 
\end{definition}

     \begin{theorem}\label{thm:Deltasameeps}
        Let $(X, \pi)$ be a $(\dimension+1)$-partite (weighted) totally connected complex.  For some $0< \delta <1$, assume that  
        \begin{align*} \gamma_2 \leq \frac{\delta^2}{10 (1+\ln \Delta)} \quad \text{and} \quad  \gamma_2
        \leq \frac{1- \delta}{\Delta +\ln \Delta}.
        \end{align*}
        Then,  the link of any face $\tau$ of co-dimension $k$  of $X$ has spectral expansion 
        $$\begin{cases}
        \frac{c(1-\delta)}{k\delta} & \text{if } k\geq \Delta,\\
        \frac{c(1-\delta)\frac{k+\ln k}{\Delta+\ln\Delta}}{k\delta}&\text{if } k<\Delta,
        \end{cases}
        $$
        for some constant $c \leq 2$ that depends on $\delta$. 
    \end{theorem}
    Note that, for $\Delta = d$, this theorem retrieves  \cref{thm:trickledown} up to a lower order term in the condition on $\gamma_2$ and a constant in the bounds on local spectral expansions. 
    When  $\Delta \ll d$, this theorem is a significant improvement over \cref{thm:trickledown}. 
    Roughly speaking, this theorem says that, if the complex has many faces of co-dimension 2 whose links are 0-expanders, one needs to satisfy a much weaker condition on $\gamma_2$ to get $O(1/k)$-spectral expansion for faces of co-dimension $k$. In other words, the faces of co-dimension 2 that have perfect spectral expansion can compensate for faces of co-dimension 2 that have bad spectral expansion. 

    Next, we state the second special case of our theorem. For every integers $1 \leq n$, let $H_n = \sum_{i=1}^n \frac{1}{i}$ be the n-th harmonic number. Moreover, for any $1 \leq i \leq n$ define $H_n(i) =   \sum_{j = i}^n \frac{1}{j}$ and let $H_n (0) = H_n(1)$.

     \begin{theorem}
        Let $(X, \pi)$ be a $(\dimension+1)$-partite (weighted) totally connected complex. For any distinct $\ti, \tj \in \prtin$, let $\epsilon_{\{\ti, \tj\}} = \max_{\tau: \type(\tau)=\prtin \setminus \{\ti, \tj\}} \lambda_2(P_\tau)$ 
        be the  2nd largest eigenvalue of the simple random walk matrices on  $(X_\tau, \pi_\tau)$ for all  $\tau$ of type $\prtin \setminus \{\ti, \tj\}$. 
        For some $0< \delta <1$, assume that  for every $\ti \in \prtin$,  
        \begin{eqnarray*} \epsilon_{\{\ti, \tj\}} \cdot H_{d} &\leq& \frac{\delta^2}{10},  \forall \tj \neq \ti  \quad \text{and}  \\ \sum_{ \ell = 1 }^{d}  \epsilon_{\{\ti, \tj_\ell\}}  \cdot \frac{H_{d} (\ell -1)}{d}
        &\leq& \frac{1- \delta}{d},
        \end{eqnarray*}
        where $\tj_0 \dots, \tj_\dimension$ is an ordering of $\prtin \setminus i$ such that $\epsilon_{\{\ti, \tj_0\}} \leq \dots \leq \epsilon_{\{\ti, \tj_\dimension\}}$. 
        Then,   $X$ is  $(\frac{c(1-\delta)}{\delta},  \dots,\frac{c(1-\delta)}{ \dimension \delta})$-local spectral expander for some constant $c \leq 2$ that depends on $\delta$.  
    \end{theorem}
  We remark that for every any $\ti \in \prtin$, $1 \leq \frac{\sum_{\ell=1}^{d} H_{d}(\ell-1)}{d} \leq 1+\frac{\ln d}{d}$. So, roughly speaking, the latter condition can be seen  as  $\mathbb{E}_j [\eps_{\{i,j\}}] \leq \frac{1-\delta}{d}$ for every $\ti \in \prtin$, where the expectation is weighted according to  $\frac{H_d(.)}{d}$. This is an improvement over the stronger condition in \cref{thm:trickledown}. Now, we state the main theorem. 
  
      \begin{theorem}[Main] \label{thm:main_technical}
        Let $(X, \pi)$ be a $(\dimension+1)$-partite (weighted) totally connected complex. For any distinct $\ti, \tj \in \prtin$, let $\epsilon_{\{\ti, \tj\}} = \max_{\tau: \type(\tau)=\prtin \setminus \{\ti, \tj\}} \lambda_2(P_\tau)$ 
        be the  2nd largest eigenvalue of the simple random walk matrices on  $(X_\tau, \pi_\tau)$ for all  $\tau$ of type $\prtin \setminus \{\ti, \tj\}$. 
        For some $0< \delta <1$, assume that  for every $\ti \in \prtin$,  
        \begin{eqnarray} \epsilon_{\{\ti, \tj\}} \cdot H_{\Delta -1} &\leq& \frac{\delta^2}{10},  \forall \tj \neq \ti  \quad \text{and} \label{eq:thmcond1} \\ \sum_{ \ell = 1 }^{\Delta(\ti)}  \epsilon_{\{\ti, \tj_\ell\}}  \cdot H_{\Delta(i)-1} (\ell -1)
        &\leq& 1- \delta,\label{eq:thmcond2} 
        \end{eqnarray}
        where $\tj_0 \dots, \tj_\dimension$ is an ordering of $\prtin \setminus i$ such that $\epsilon_{\{\ti, \tj_0\}} \leq \dots \leq \epsilon_{\{\ti, \tj_\dimension\}}$. 
        Then,  (the link of the emptyset of) $X$ is  a $\frac{c(1-\delta)}{ \dimension \delta}$-expander
        for $c=\frac{2(1+ \frac{\delta^2}{10})}{(1+\delta)}$.
    \end{theorem}
    \begin{remark}
    If, for some $\delta >0$, the conditions of the above theorem hold for a complex $(X, \pi)$, then the conditions also hold for the same $\delta$ for all links $(X_\tau, \pi_\tau)$ (of faces of co-dimension at least $2$). Therefore, this theorem implies that $X$ is  $(\frac{c(1-\delta)}{\delta},  \dots,\frac{c(1-\delta)}{ \dimension \delta})$-local spectral expander for $c=\frac{2(1+ \frac{\delta^2}{10})}{(1+\delta)}$. One can prove tighter bounds if they apply this theorem to any link $(X_\tau, \pi_\tau)$ individually and possibly use better bounds on $\Delta_{(X_{\tau},\pi_\tau)}(i)$. 
    \end{remark}
\begin{proof}[Proof of \cref{thm:Deltasameeps}]
Fix a face $\tau$ of co-dimension $k$. For brevity we abuse notation and write $\Delta_\tau$ denote $\Delta_{(X_\tau,\pi_\tau)}$. If $k\geq \Delta$ the statement follows from the above remark. In particular,  for any $i,j \in \prtin$  
\begin{align*}&\eps_{\{i,j\}}\cdot H_{\Delta_\tau-1}\leq \gamma_2\cdot H_{\Delta-1}\leq \gamma_2\cdot (1+\ln\Delta)\leq \frac{\delta^2}{10}, \\
&\sum_{\ell=1}^{\Delta_\tau(i)} \eps_{\{i,j_\ell\}}\cdot H_{\Delta_\tau(i)-1}(\ell-1)\leq \gamma_2 (\Delta+\ln\Delta)\leq 1-\delta.
\end{align*}
So, we can apply \cref{thm:main_technical}.

Otherwise, to bound the spectral expansion of $(X_{\tau} ,\pi_\tau)$,  let $\delta_k=1-(1-\delta)\frac{k+\ln k}{\Delta+\ln \Delta}\geq \delta$. 
For $i,j\in \prtin$
\begin{align*} &\eps_{\{i,j\}}\cdot H_{\Delta_\tau-1}\leq \gamma_2\cdot H_{k-1} \leq \frac{\delta^2\cdot H_{k-1}}{10(1+\ln\Delta)}\underset{\delta\leq\delta_k}{\leq} \frac{\delta_k^2}{10},\\ 
&\sum_{\ell=1}^{\Delta_\tau(i)}\eps_{\{i,j_\ell\}}\cdot H_{\Delta_\tau(i)-1}\underset{\eps_{i,j_\ell}\leq\gamma_2,\Delta_\tau(i)\leq k}{\leq} \gamma_2(k+\ln k)\leq \frac{(1-\delta) (k+\ln k)}{\Delta+\ln\Delta}=1-\delta_k.
\end{align*}
Therefore, applying  \cref{thm:main_technical} to $(X_\tau,\pi_\tau)$, we obtain that $(X_\tau,\pi_\tau)$ is a $\frac{c(1-\delta_k)}{k\delta}$-expander.
\end{proof}

\paragraph{Applications to Graph Coloring.} Consider a graph $G= ([n], E)$ with degree function $\Delta: [n] \rightarrow \mathbb{Z}_{\geq 0}$ and maximum degree $\Delta$, paired with a collection of color lists $\{L(i)\}_{i \in [n]}$ satisfying $L(i) \geq \Delta(i) + (1+\eta) \Delta $ for all $i \in [n]$ and for some $0<\eta\leq 0.9$ such that  $\frac{1+\ln \Delta}{\Delta} \leq \frac{  \eta^2, }{40}$. We  define the $(n+1)$-partite coloring complex $X(G, L)$ specified by the following facets: $\{i, \sigma (i)\}_{i \in [n]}$ is a facet if and only if $\sigma$ is a proper $L$-coloring of $G$, i.e. $\sigma (i) \in L(i)$ for each $i \in [n]$ and  $\sigma (i) \neq \sigma (j)$ if $ \{i, j\} \in E$. It is not hard to see that if $\{i, j\} \notin E$, then $\epsilon_{\{i, j\}} =0$. Moreover, if $\{i, j\} \in E$, then $\epsilon_{\{i, j\}} \leq \frac{1}{(1+\eta)\Delta}+ \frac{1}{(1+\eta)^2\Delta^2}$ (see Theorem 4.4 in \cite{ALO21}). Once can verify that if we apply the above theorem  to the coloring complex $X(G, L)$ with $\delta = \frac{\eta}{2}$, we get that  $X(G, L)$ is a $\left(\frac{4}{\eta},\frac{4}{2\eta },  \dots, \frac{4}{ (|V|-1 )\cdot \eta}\right)$-local spectral expander, and thus the Glauber dynamics for sampling a random  proper coloring mixes in polynomial time. This retrieves (up to constants) a theorem proved in \cite{ALO21}. 


\paragraph{Applications to Sparse High Dimensional Expanders}

Kaufman and Oppenheim \cite{KO18} obtained a simple construction of sparse $(\dimension+1)$-partite complexes with $|X(0)|\geq p^s$  for any integer $s >\dimension$ and prime power $p$ such that  every  $x\in X(0)$ is in at most $p^{O(d^3)}$ many facets (hence the degree is independent of $s$). 
They argued that for any non-consecutive pair of parts $i,j\in [d]$, i.e., $j\neq i+1$ and $i\neq j+1$ (mod $d+1$), we have $\eps_{\{i,j\}}=0$ but $\eps_{\{i,i+1\}}\leq \frac{1}{\sqrt{p}}$ for any $i\in [d]$ ($i+1$ is taken modulo $d+1$). Consequently, $\Delta(i)=2$ for any $i\in [d]$. Then, using  \cref{thm:trickledown}, they show that the complex is a $(\frac{1}{\sqrt{p}-(d-2)}, \dots,\frac{1}{\sqrt{p}-d-2})$-local spectral expander for $p> (d-2)^2$.
Simply plugging in these values into the above theorem, for $\delta=1-\frac{2}{\sqrt{p}}$ and  $p\geq 193$ (independent of $d$) the assumptions of the theorem are satisfied. The resulting complex is $(\frac{2c}{\sqrt{p}\delta},\dots,\frac{2c}{d\sqrt{p}\delta})$-local spectral expander for $c\approx 1.15$.
In other words, not only does the Kaufman-Opennheim construction give a HDX for constant values of $p$ independent of $d$, but also its local spectral expansion improves inverse linearly with the co-dimension.

O'Donnell and Pratt \cite{OP22} constructed $(d+1)$-partite  (sparse) high-dimensional expanders, with unbounded dimension $d$, via root systems of simple Lie Algebras, namely families $A_d$ for $d\geq 1$, $B_d$ for $d\geq 2$, $C_d$ for $d\geq 3$ and $D_d$ for $d\geq 4$. For explicit descriptions of these root systems, see e.g. \cite[Sec. 3.6]{Car89}.
O'Donnell and Pratt \cite{OP22} showed that, similar to the Kaufman-Oppenheim construction, the resulting $d$-dimensional complex $X$  satisfies $|X(0)|\geq p^{\Theta(m)}$ whereas every vertex is only in $p^{\Theta(d^2)}$ many facets and for any $i,j\in [d]$, $\eps_{i,j}\leq \sqrt{2/p}$. Then, using \cref{thm:trickledown} they concluded that the complex is a $(\frac{1}{\sqrt{p/2}-d+1},\dots,\frac{1}{\sqrt{p/2}-d+1})$-local spectral expander.
Upon further inspection of the explicit set of roots, one can verify that $\Delta\leq 2$ for complexes based on $A_d,B_d,C_d$ root systems and $\Delta\leq 3$ for the $D_d$ root system. Plugging in these values in the above theorem and setting $\delta=1-2\sqrt{2/p}$ for $A_d,B_d,C_d$ complexes and $\delta=1-3.5\sqrt{2/p}$ for the $D_d$ complex, if $p\geq 376$ for $A_d,B_d,C_d$ complexes and  $p\geq 729$ for the $D_d$ complex, we get that these complexes are  $(\frac{c'}{\sqrt{p}\delta}, \dots,\frac{c'}{d\sqrt{p}\delta})$-local spectral expander for some constant $c'>1$.

The well known Ramanujan complexes, also known as LSV complexes, are generalizations of Ramanujan graphs that were introduced by Lubotsky, Samuels, and Vishne in \cite{LSV05a} and explicitly constructed in  \cite{LSV05b}. Any $d$-dimenssional LSV complex $X$ that is $q$-thick for some fixed prime power $q$ and $d\geq 2$ has a bounded degree (the number of facets that contain each $x \in X(0)$ only deponents on $q$ and $d$, and is constant in the size of the ground set $n$ which can be arbitrarily large) (e.g. see \cite{EK16}). Moreover,  the link of every proper face of type $S$ is a spherical building complex in which  $\Delta(i) = |\{j \neq i: \epsilon_{\{i,j\}}>0\}|$ is at most 2 for every $i \in \prtin \setminus S$. Furthermore, the worst expansion among links co-dimension 2 is  $\frac{c}{\sqrt{q}}$, for some constant $c$ independent of $q, d, n$.   So, there is a constant $q_0$ such that if $q \geq q_0$, \cref{thm:main_technical} implies that the link of any (proper) face  of $X$ of co-dimension $k$ is a $\frac{c'}{(k-1)\sqrt{q}}$-spectral expander  for some constant $c'>0$ independent of $q, d, n$. This  improves over the bound $\frac{C(d)}{\sqrt{q}}$  proved in \cite{EK16}, where $C(d) \geq 2^d (d+1)!$. 



\subsection{Proof Overview}
At a high-level, our proof builds on the matrix trickle down framework introduced in the work of the authors with Liu \cite{ALO21}.  The Oppenheim's trickle down theorem follows from an inductive argument that derives a bound on the second eigenvalue of the simple walk on 1-skeleton of each link $(X_\tau, \pi_\tau)$  using the largest second eigenvalue of the simple walk on the 1-skeleton of links $(X_{\tau'}, \pi_{\tau'})$ for all faces $\tau' \supset \tau$ of size $|\tau|+1$. The reason that one has to take the largest 2nd eigenvalue as opposed to the average in each inductive step is that the eigenspaces of these simple walks are very different. The matrix trickle down framework overcomes this issue by substituting the scalar bounds on the second eigenvalues with matrices that upper bound the transition probability matrices of the simple walks on the 1-skeletons of links. However, as opposed to Oppenhiem's trickle down theorem, the matrix trickle down framework cannot be applied in a black-box manner to bound the spectral expansion of the 1-skeletons of all links only by bounding the spectral expansion of the 1-skeletons of links of faces of co-dimension 2. The main result of this paper can be seen as applying the matrix trickle down framework with a carefully chosen set of upper-bound matrices to prove an improved trickle down theorem for partite complexes that can be applied in the same black-box fashion, just known an "average" second eigenvalue.  

Our technical contribution in this paper are twofold: First, we observe that for any two disjoint sets of parts $S,T\subseteq [d]$, if the links of all faces of co-dimension 2 whose types intersect with both $S,T$ are $0$-spectral expanders, then for any $\sigma\in X$ of type $S$ and $\tau$ of type $T$ we get
$$\P_{\eta\sim\pi}[\sigma\subset \eta | \tau\subset\eta]=\P_{\eta\sim\pi}[\sigma\subset \eta]\quad \text{ and } \quad \P_{\eta\sim\pi}[\tau\subset\eta | \sigma\subset\eta]=\P_{\eta\sim\pi}[\tau\subset\eta],$$
namely, the conditional distributions on these types are independent (see \cref{lem:main_product} for details). This observation significantly simplifies invoking the Matrix trickle down framework. Armed with this tool, we invoke the matrix trickle down theorem using a carefully chosen family of (diagonal) matrices as the matrix bounds. These matrices are recursively defined based on an "average" of the spectral expansions of the links of all faces of co-dimension 2, See the proof of \cref{thm:main_technical} for the construction of these matrices. 

\subsection{Acknowledgments}
The discussion that initiated this work took place at the DIMACS Workshop on Entropy and Maximization. We would like to thank the DIMACS center and the workhop organizers for making this happen.
In particular, we would like to thank Tali Kaufman  for raising the question of an improved trickle down theorem for sparse simplical complexes in that workshop. 
We also would like to thank Ryan O’Donnell and Kevin Pratt for helpful discussions on high dimensional expanders based on Chevalley groups. 

   \section{Preliminaries}
   
    For any integer $n \geq 0$,  let  $[n] \coloneqq \{0, \dots, n\}$.   When it is clear from context, we write $x$ to denote a singleton $\{x\}$. 
    Given a set $S$, we write $v \in \mathbb{R}^S$ and $M \in \mathbb{R}^{S \times S}$ to respectively  denote  a vector and a matrix indexed by $S$.  Given a  probability  distributions $\mu$ over a set $S$,   we  may view $\mu$ as a vector  in $ \mathbb{R}_{\geq 0}^S$. For a $n \times n$ matrix $M$ with eigenvalues $\lambda_1,\dots,\lambda_n$,  define $\rho(A) \coloneqq \max_{1\leq i\leq n}|\lambda_i|$. 
    \paragraph{Graphs}
    Given a graph $G = (V, E)$, for any $v \in V$, let $\Delta_{G}(v) $ be the degree of $v$ in $\graph$, and let  $\Delta_{G}$ be the maximum degree of $\graph$.  Moreover, given a subset $S \subseteq V$,  $G[S]$ denotes  the induced subgraph of $G$ on the set of vertices $S$. For any $ S \subseteq V$, define $\graph_{S} \coloneqq \graph[V \setminus S]$. For simplicity of notation, when $G$ is clear from context, we denote  $\Delta_{\graph}(v)$ by $\Delta (v)$ for any $v \in V$, and  for any $S \subseteq V$, we denote $\Delta_{\graph_S} (v)$  by $\Delta_S (v)$  for any $v \in V \setminus S$. Similarly, we denote the maximum degree of $G$ and $G_S$ by $\Delta$ and  $\Delta_{S}$ respectively. Moreover, when $G$ is clear from context, we write $u \sim v$ if  $u, v$ are  adjacent vertices in $G$ and  $u \sim_S v$ if $u, v \in V \setminus S$ and $u \sim v$. 
    
    We say that  a graph $G = (V, E)$ paired with a  weight function $w: E \rightarrow \mathbb{R}_{\geq 0}$ is $\epsilon$-expander if $\lambda_2 (P) \leq \epsilon$, where $P \in \mathbb{R}^{V \times V}$ is the transition probability matrix of the simple random walk on $(G, w)$  defined as   $P (x, y) = \frac{w(\{x, y\})}{\sum_{z} w(\{x, z\})}$ for any $x, y \in V$. For such a graph we write $d_w(x)=\sum_{y\sim x} w(\{x,y\})$ to denote the weighted degree of a vertex $x$ and $\vol(S)=\sum_{x\in S} d_w(x)$ to denote the volume of a set $S\subseteq V$.
    
  \subsection{Linear Algebra}
  \begin{lemma}[Cheeger's Inequality]\label{lem:cheeger}
  For any  graph $G=(V,E)$ with weights $w:E\to\R_{\geq 0}$ and any 	$S\subseteq V$,
  $$ \frac{w(E(S,\overline{S}))}{\min\{\vol(S),\vol(\overline{S})\}}\leq \sqrt{2(1-\lambda_2)}$$
  where $\lambda_2$ is the second largest eigenvalue of the simple random walk on $G$
  \end{lemma}

  \begin{fact}[Expander Mixing Lemma]\label{fact:expmixinglemma}
Given a (weighted) graph $G=(V,E,w)$, for any set $S\subseteq V$,
$$ \left|w(E(S)) - \frac{\vol(S)^2}{\vol(V)}\right|\leq \lambda_2 \vol(S),$$
where $\lambda_2$ is the second largest eigenvalue of the simple random walk on $G$.
\end{fact}

   \subsection{Simplicial Complexes} 
 We say that a simplicial complex $X$ is {\em gallery connected} if for any face $\tau$ of co-dimension at least 2 and any pair of facets $\sigma, \sigma'$ of $X_\tau$ there is a sequence of facets of $X_\tau$, $\sigma=\sigma_0,\sigma_1,\dots,\sigma_\ell=\sigma'$, such that for all $0\leq i<\ell$, $|\sigma_i\Delta \sigma_{i+1}|=2$. 
It is shown in \cite[Prop 3.6]{Opp18} that if $X$ is totally connected, then it is gallery connected. 
    \begin{lemma} \label{lem:connectivity}
        Consider  a totally connected $(\dimension+1)$-partite complex $X$ with parts indexed by $\prtin$. For any $S \subseteq \prtin$, The induced subgraph of the 1-skeleton of $X$ on vertices of type $S$ is connected. 
    \end{lemma}
    \begin{proof}
        Take $x, y$ of  type $\ti,\tj \in S$ and facets $\eta, \eta'$ such that $x \in \eta$, $y \in \eta'$. Total connectivity implies that there is a  sequence $\eta = \eta_1, \dots, \eta_t  = \eta'$ such that $\eta_i \cap \eta_{i+1} \neq \emptyset$ for all $1\leq i \leq t-1$. Let $\sigma_1 \subseteq \eta_1, \dots, \sigma_t \subseteq \eta_t$ be faces of type $\{\ti, \tj\}$. Then $\sigma_1, \dots, \sigma_t$ gives a path between $x, y$. 
   \end{proof}

 Given a (weighted) complex $(X,\pi)$,
for integer $-1\leq i \leq \dim(X)-1$, $\pi$ induces a distribution  $\pi_{i}$ on $X(i)$, 
    \begin{align*}
        \pi_i(\sigma) = \frac{1}{{\dim(X)+1  \choose i+1}} \Pr_{\tau \sim \pi}[  \sigma  \subset \tau] \quad \forall \sigma \in X (i). 
    \end{align*}
    Let   $ P_{(X, \pi), \tau} \in  \mathbb{R}^{X(0) \times X(0)}$  denote the transition probability  matrix of the simple random walk  on the 1-skeleton of $(X_\tau, \pi_\tau)$ padded with zeros outside the $X_\tau(0) \times X_\tau(0)$ block, i.e.    $P_{(X, \pi), \tau} (x, y) = \frac{\P_{\sigma\sim\pi_\tau}[\{x, y\}\subset\sigma]}{\sum_{z \in x_\tau (0) } \P_{\sigma\sim\pi_\tau}[\{x, z\}\subset \sigma]}$ for  $x, y \in  X_\tau(0)$,  and  $P_\tau (x, y) = 0$ otherwise. When the  weighted complex $(X,\pi)$ is clear from context, we write $P_{\tau}$ to denote $P_{(X, \pi), \tau}$. For any $\tau$ of co-dimension at least 2, we define the diagonal matrix $\Pi_{(X, \pi), \tau} \in \mathbb{R}^{X(0) \times X(0)}$ as follows:   $\Pi_{(X, \pi), \tau} (x, x ) = \pi_{\tau, 0} (x)$  for $x \in X_\tau (0)$, and $\Pi_{(X, \pi), \tau} (x, x ) =0$ otherwise. When $(X,\pi)$ is clear from context, we write $\Pi_{\tau}$ to denote $\Pi_{(X, \pi), \tau}$. Note that $\Pi_\tau P_\tau$ is a symmetric matrix. 
    
   Given a  $(\dimension+1)$-partite  complex,
   we say that an $x \in X(0)$  is of type $i$ and write  $\type (x) = i$ if $x \in T_i$. Similarly, the type of a face $\tau \in X$ is defined as $\type (\tau) = \{\ti \in \prtin : |\tau \cap \prt_\ti|=1\}$. The following facts hold for weighted partite complexes.

       \begin{fact} \label{fact:partite_normalize}
        Consider  a weighted $(\dimension+1)$ partite complex $(X, \pi)$ and a face $\tau$ of co-dimension $k \geq 1$. We have $k  \pi_{\tau, 0} (x)  =  \Pr_{\sigma \sim \pi_\tau}[x \in \sigma ]$ for all $x \in X_\tau (0)$.
    \end{fact}
    
    \begin{fact} \label{fact:partite_sum_to_1}
        Consider  a weighted $(\dimension+1)$ partite complex $(X, \pi)$ with parts  indexed by $\prtin$ and a face $\tau$ of co-dimension $k \geq 1$.  For any $\ti \in \prtin$,  $ \sum_{x:  \type(x) =i} \Pr_{\sigma \sim \pi_\tau}[x \in \sigma ]= 1$. 
    \end{fact}
The following definition is useful for proving the main theorem. 
    \begin{definition}
        For any $(\dimension+1)$-partite complex $(X, \pi)$ with parts  indexed by $\prtin$, define a graph $\graph_{(X, \pi)}$ on the set of vertices $\prtin$, where any distinct $\ti, \tj \in \prtin$ are   adjacent in $\graph_{(X, \pi)}$ if there exists $\tau$ of type $\prtin \setminus \{\ti, \tj\}$  such that the second eigenvalue of $(X_\tau, \pi_\tau)$ is positive.
    \end{definition}
    \begin{remark}
        For any $(\dimension+1)$-partite complex $(X, \pi)$ with parts  indexed by $\prtin$, for every $i \in \prtin$, $\Delta(i)$ (see \cref{def:Delta(i)}) is the degree of $i$ in graph $G_{(X, \pi)}$ and $\Delta$ is the maximum degree of $G_{(X, \pi)}$. 
    \end{remark}
     Note that if $\codim(\tau) =k$, the link $X_\tau$ is a $k$-partite complex with parts indexed by $\prtin \setminus S$.  One  can verify that given a face $\tau$ of type $S$,  the set of edges of $G_{(X_\tau, \pi_\tau)}$ is  a subset of the edges of $(G_{(X, \pi)})_S$, i.e., the induced subgraph of $G_{(X, \pi)}$ on $\prtin \setminus S$. When $(X, \pi)$ is clear from context, we write $G$ for $G_{(X, \pi)}$ and $G_S$ for $(G_{(X, \pi)})_S$.

    \paragraph{Product of Weighted Complexes }
    Given  weighted complexes  $(Y_1,  \mu_{1}),   \dots,  (Y_\ell,  \mu_{\ell})$   defined on disjoint ground sets and  of dimensions $\dimension_1, \dots, \dimension_\ell$ respectively, and a  weighted  complexes $(X, \pi)$ of  dimension $\dimension$, 
    we write $(X, \pi) = (Y_1, \mu_1) \times \dots \times (Y_\ell,  \mu_\ell)$ if  $X (\dimension) = \{\cup_{i \in [\ell]} \tau_i:   \tau_1 \in Y_1  (\dimension_1), \dots, \tau_\ell \in Y_\ell (\dimension_\ell) \}$  and $\pi(\cup_{i \in [\ell]} \tau_i) = \prod_{ i \in [\ell]}  \mu_i (\tau_i) $ for all $\tau_1 \in Y_1  (\dimension_1), \dots, \tau_\ell \in Y_\ell (\dimension_\ell)$. We denote the generating polynomial of $(X, \pi)$ by    $g_{(X, \pi)}$, i.e.   $g_{(X, \pi)} \coloneqq \sum_{\tau \in X(\dimension) }  \pi (\tau) \prod_{x \in \tau} z_x$. One can verify that 
     $(X, \pi) = (X_1, \mu_1) \times \dots \times (X_\ell,  \mu_\ell)$ if and only if  $g_{(X, \pi)} = g_{ (X_1, \mu_1)} \times \dots \times g_{(X_\ell,  \mu_{\ell})}$. Note that this is true because  we assume that for any weighted simplicial complex, the given distribution on facets is non-zero on all facets. 

    \paragraph{Matrix Trickle Down Theorem }
    We use the following theorem which is the main technical theorem in \cite{ALO21}. 
    
    \begin{theorem}[{\cite[Thm III.5]{ALO21}}]\label{thm:matrix-trickle-down}
        Let $(X, \pi)$ be a totally connected weighted complex. Suppose $\{M_{\tau} \in \mathbb{R}^{X(0) \times X(0)}\}_{\tau \in X(\leq \dimension-2)}$ is a family of symmetric matrices satisfying the following:
        \begin{enumerate}
        \item \textbf{Base Case:} For every $\tau$ of co-dimension 2, we have the spectral inequality $$\Pi_{\tau}P_{\tau} - 2\pi_{\tau, 0}\pi_{\tau, 0}^{\top} \preceq M_{\tau} \preceq \frac{1}{5} \Pi_{\tau}.$$
        \item \textbf{Recursive Condition:} For every $\tau$ of co-dimension at least $k \geq 3$, at least one of the following holds:  $ M_{\tau}$ satisfies 
        \begin{align}    \label{eq:desired_ineq}
            \boundMatrix_\tau \preceq  \frac{k-1}{3k -1} \Pi_\tau \quad \text{and}  \quad
            \E_{x \sim \pi_{\tau}} M_{\tau \cup \{x\}} \preceq M_{\tau} - \frac{k-1}{k-2} M_{\tau}\Pi_{\tau}^{-1}M_{\tau}.
        \end{align}
        Or,  $(X_\tau,\pi_{\tau, k-1})$ is a product of weighted simplicial complexes $(Y_{1},\mu_{1}),\dots,(Y_{t},\mu_{t})$ and  for every $\eta \in X_{\tau}(k-1)$,
        \begin{align*}
            M_{\tau} = \bigoplus_{1 \leq i \leq t : \dimension_{Y_i} \geq 1} \frac{\dimension_{Y_i}(\dimension_{Y_{i}}+1) }{k(k-1)} M_{\tau \cup \eta_{-i}},
        \end{align*}
        where $\eta_{-i} = \eta \setminus Y_{i}(0)$.
        \end{enumerate}
        Then for every $\tau \in X(\leq \dimension-2)$, we have the bound 
        $\lambda_2(\Pi_\tau P_\tau)\leq \rho(\Pi_\tau^{-1}M_\tau)$.
    \end{theorem}
        
    \section{Simplifying  Matrix Trickle Down's Conditions to Scalar Inequalities}

   In this section, given a $(\dimension+1)$-partite complex $ (X, \pi)$, we apply the matrix trickle down theorem to derive a set of conditions on a family of vectors $\{\vect_S \in \mathbb{R}^{\prtin}\}_{S \subset \prtin, |S|< \dimension}$ that will guarantee that $\lambda_2(P_\tau) \leq \frac{\max_{\ti \in \prtin} \vect_S(\ti)}{k-1}$ for all $k\geq 2$ and $\tau$ of co-dimension $k$ and type $S$.  We prove the following theorem. 

    \begin{theorem}\label{prop:Dcondition}
        Consider a totally connected $(\dimension+1)$-partite complex $ (X, \pi)$ with parts indexed by $\prtin$ and graph $\graph \coloneqq \graph_{(X, \pi)}$. Suppose we are given a family of vectors   $\{\vect_S \in \mathbb{R}^{\prtin}\}_{S \subset \prtin, |S|< \dimension}$ such that for all   $S \subset \prtin$ of size $(\dimension +1)-k$, the support of $f_S$ is a subset of $\prtin \setminus S$, and the following holds: 
        \begin{itemize}
      
      
            \item           If $G_S$ is disconnected, then $\vect_S = \sum_{1 \leq i \leq \ell: |I_i| \geq 2}  \vect_{[d] \setminus I_i},$ where $I_1 \cup \dots \cup I_\ell $ are the vertices of the  connected components of $G_S$. Note that if all connected components are of size 1, then $f_S = 0$. 

            \item Otherwise if $G_S$ is connected,   we have $\max_{\ti \in \prtin} \vect_S (\ti) \leq \frac{(k-1)^2}{3k -1}$ and 
        \begin{enumerate}[label=(\roman*)]
            \item  Base Case: If  $k = 2 $,  then  for  every face $\tau$ of type $S$,  
            $\lambda_2(P_\tau) \leq   \max_{i \in \prtin \setminus S}\vect_{S} (\ti)$.  
            \label{item:cond1}
            \item  Recursive Condition: If $k \geq 3$, then 
                $$\sum_{\tj \in \prtin \setminus (S \cup \ti)}   \vect_{S \cup \tj} (\ti) \leq  (k-2) \vect_S (\ti)- \vect_S^2 (\ti), $$
               \label{item:cond2}
            for all $i \in \prtin \setminus S$. 
         \end{enumerate}
             \end{itemize}
         Then, for all $k\geq 2$ and $\tau$ of  co-dimension $k$ and type $S$, $\lambda_2(P_\tau) \leq \frac{\max_{\ti \in \prtin} \vect_S(\ti)}{k-1}$. 
         
    \end{theorem}

    The main sets of conditions in the above theorem are the inequalities in \cref{item:cond1} and \cref{item:cond2}. To get some intuition about these conditions, it is helpful to compare the above with the standard trickle down theorem (\cref{thm:trickledown}). There, one  shows that if $\lambda_2(P_{\tau\cup \{x\}})\leq \lambda$ for all  $x\in X_\tau (0)$, then $\lambda_2(P_\tau)\leq \alpha$, where satisfies \begin{equation}\label{eq:tricklerecursion}\lambda\leq \alpha-{\alpha}^2(1-\lambda).\end{equation} Then, \cref{thm:trickledown} follows by recursively applying this inequalities.
    
In the above theorem, instead of a single upper bound on $\lambda_2(P_\tau)$ for faces $\tau$ of co-dimension 2,   one bounds the expansion of the links of all faces of co-dimension 2 of each type separately, allowing higher degrees of freedom. 
For any face $\tau$ of type $S$ and co-dimension $k=|S|$, the function $\frac{f_S(.)}{k-1}$ will serve as the digonal entries of a matrix upper-bound $P_\tau$.
Then, the inequality $\frac{\sum_{\tj \in \prtin \setminus (S \cup \ti)}   \vect_{S \cup \tj} (\ti)}{k-2} \leq  \vect_S (\ti)- \frac{\vect_S^2 (\ti)}{k-2}$ 
is the natural analogue of \eqref{eq:tricklerecursion}
which requires $f_S$ to be at least  ``the average'' of $f_{S \cup j}$ for all $j \in \prtin \setminus S$ plus an square error term.     

   Before proving the above theorem, we show that if $G_S$ is disconnected with parts $G[I_1], \dots, G[I_\ell]$ for some $S\subset \prtin$ of size  at most $d-1$, then for any $\tau$ of type $S$, $(X_\tau,\pi_{\tau, k-1})$ can be written as a product of family of its links of  types  $\prtin \setminus I_i$ for all $1 \leq i \leq \ell$. This allows us to prove a better upper-bound on $\lambda_2(P_\tau)$ for such faces $\tau$ by simply "concatenating" upper-bounds on each connected component of $G_S$.

    \begin{lemma}\label{lem:pairwiseindependence}
        Consider a  2-partite complex $(X, \pi)$ with parts $S, T$. If $(X, \pi)$ is 0-expander, then $(X, \pi) = (X_{z}, \pi_{z}) \times  (X_{y}, \pi_{y }) $ for any $y \in S$ and $z \in T$. 
    \end{lemma}
    \begin{proof}
        Note  that $(X, \pi)$ is a weighted bipartite graph with parts $S, T$. Let $A \in \mathbb{R}^{X(0) \times X(0)}$ be the adjacency matrix of $(X, \pi)$. Let $A_{S,  T} (y, z) = A (y, z)$  for $y \in S$, $z \in T$ and $0$ on other entries.
        Moreover, let $A_{T,  S} = A - A_{S,  T}$.  Then, for any vector $v \in \mathbb{R}^{X(0)}$, we get $A = A_{S,  T} v_T + A_{T, S} v_S $, where $v_S, v_T$ are respectively supported on $S, T$  and $v = v_S +v_T$. Thus, if $A v = \lambda v$, then  $A v' = -\lambda v'$, for $v' = (-v_S + v_T)$. So if $\mu$ is an eigenvalue of $A$, then $-\mu$ is also an eigenvalue of $A$. Thus, if    $(X, \pi)$ is 0-expander, the rank of $A$ is 2. This implies that  there are vectors $w_S \in \mathbb{R}^{S}$ and $w_T \in \mathbb{R}^{T}$ such that 
        $\pi(\{y, z\}) = A(y, z) = A (z, y) = w_S(y) w_T(z)$  for $y \in S$, $z \in T$. Without loss of generality, assume $\|w_S\|_1=\|w_T\|_1 =1 $. Then, for any $y \in S$ and $z \in T$, we have $\pi_{z} (y) = \frac{\pi (\{y, z\})}{\sum_{x \in S} \pi(\{x, z\})} = w_S (y)$. Similarly $\pi_{y}(z) = w_T(z)$.  Thus  $\pi(\{y, z\}) =  \pi_{y}(z) \pi_{z}(y)$. This finishes the proof. 
    \end{proof}
        
    \begin{lemma}\label{lem:main_product}
        Consider a totally connected $(\dimension+1)$-partite complex $ (X, \pi)$ with parts indexed by $\prtin$ and its associated graph $G \coloneqq G_{(X, \pi)}$. Let $I_1 \cup \dots \cup I_\ell$ be a partition of $\prtin$ such that for any $1 \leq i \leq \ell$ the induced graph $G[I_i]$ is a  connected component or the union of several connected components of  $\graph$.  Then $(X, \pi)  = (X_{ \sigma_{-1}}, \pi_{\sigma_{-1}}) \times  \dots \times (X_{ \sigma_{-\ell}}, \pi_{\sigma_{-\ell}})$, where  $\sigma_{-i}$ is an arbitrary face of type $\prtin \setminus I_i$ for any $1 \leq i \leq \ell$. 
    \end{lemma}
        \begin{proof}
        We prove the statement by induction on $\dimension$. For $ \dimension = 1$,  the statement simply follows from \cref{lem:pairwiseindependence}. Now, assume that $\dimension>1$.	If $|I_i| =  1$  for all $1 \leq i \leq \ell$, then $\ell \geq 3$. In this case,  let  $S \coloneqq  I_1 \cup I_2$. Otherwise, WLOG assume that $|I_1| \geq 2$ and let $S \coloneqq I_1$. First, we show that $g_{(X, \pi)}$ can be written as $g_{(X, \pi)} = h \cdot h'$, where $h$ is a polynomial  in   $\{z_y :  \type(y) \in     I \setminus S \} $ and  $h'$ is a polynomial in terms of variables  in   $\{z_y : \type(y)  \in   S \} $. 
        By induction hypothesis, for any $\ti \in S$,  $x \in \prt_\ti$, and any face $\sigma \in X$ of type $S $  such that $x \in \sigma$
        \begin{equation}\label{eq:indhyp} 
            \partial_{z_x} g_{(X,\pi)} = f^x \cdot g^{x}
        \end{equation}
        where $f^x$  is a polynomial in terms of variables  in   $\{z_y :  \type(y) \in    S \setminus \ti \} $ and  $g^x$ is a polynomial in terms of variables  in   $\{z_y : \type(y)  \in   I \setminus S \} $.  
        Now, take  arbitrary    $\ti , \tj \in S$ such that $\ti\neq \tj$. Then,  \eqref{eq:indhyp} implies that for any  face $\{x, y\}$ of type $\{\ti, \tj\}$
        	$$ \partial_{z_x}\partial_{z_y} g_{(X,\pi)} =   (\partial_{z_y} f^x)  
        g^x	= (\partial_{z_x} f^y) g^y$$
        It thus follows that $g^x$ is a multiple of $g^y$. One can see this simply by substituting 1 for all variables in $\{z_y :  \type(y)  \in S \setminus \{\ti, \tj \} \}$. Moreover, since $g^x$ and  $g^y$ are  generating polynomials of distributions, i.e. the coefficients  sum up to 1, we get $g^x = g^y$. Therefore, we get that for any distinct $x, y$  such that  $\type(x), \type(y) \in S$ and $\{x, y\}$ is a face,  $g^x = g^y$.  Applying \cref{lem:connectivity}, we get $g^x  = g^y $ for all $x, y \in \cup_{\ti \in   S } \prt_\ti$.    Thus, there exist a polynomial $h$ in variables $\{z_y : \type(y ) \in  I \setminus S \} $ such that we can rewrite  \eqref{eq:indhyp} for any $x$ with $\type (x) \in S$ as 
        	$$ \ \partial_{z_x} g_{(X,\pi)} = f^x \cdot h,
        	$$
        where $f^x$  is a polynomial in terms of variables  in   $\{z_y : \type(y) \in S \setminus \ti \} $. 
        Finally, since $X$ is a partite complex,
        
          \begin{equation}\label{eq:product} 
            |S| g_{(X,\pi)}= \sum_{\ti \in S} \sum_{x \in \prt_\ti} z_x  \partial_{z_x} g_{(X,\pi)} = h \cdot  \sum_{\ti \in S} \sum_{x \in \prt_\ti}  z_x f^x = h \cdot h',
        \end{equation}
         where $h' = \sum_{\ti \in S} \sum_{x \in \prt_\ti}  z_x f^x $ is a polynomial in $\{z_y : \type(y) \in  S  \} $. It remains to show that for any face $\sigma$ of type $S$, we have $h = g_{(X_\sigma, \pi_\sigma)}$, and for any $\tau$ of type $\prtin \setminus S$, we have $h'= g_{(X_\tau, \pi_\tau)}$. Fix arbitrary faces $\sigma$ of type $S$ and $\tau$ of type $\prtin \setminus S$. Noting that $g_{(X,\pi)}$ is a multiple of $h \cdot h'$, and that $h'$ is in variables associated to elements whose types are in  $S$ and  $h$ is in variables associated to elements whose types are in $\prtin \setminus S$, we conclude that $h'$ has a monomial that is a multiple of  $\prod_{x \in \sigma} z_{x}$  and $h$ has a monomial that is a multiple of  $\prod_{x \in \tau} z_{x}$. First, take $(\prod_{x \in \sigma} \partial_{z_{x}})$ from both sides of $\eqref{eq:product}$. We get that  $g_{(X_{\sigma}, \pi_\sigma)}$  is a positive multiple of $h$. Similarly, taking $(\prod_{x \in \tau} \partial_{z_{x}})$ from both sides of $\eqref{eq:product}$, we get that  $g_{(X_{\tau}, \pi_\tau)}$  is a positive multiple of $h'$. Thus, noting that the coefficients of generating polynomials sum up to 1, we get $h =  g_{(X_\sigma, \pi_\sigma)}$ and $h' = g_{(X_{\tau}, \pi_\tau)}$ as desired. Repeating the same argument inductively on the complex $(X_{\sigma}, \pi_{\sigma})$ proves the claim. 
    \end{proof}

     Now we are ready to prove \cref{prop:Dcondition}. 
    \begin{proof}[Proof of \cref{prop:Dcondition}]
        We apply \cref{thm:matrix-trickle-down}. 
        For every $S \subset \prtin$ such that $|S| <\dimension$, define a diagonal matrix $D_S \in \mathbb{R}^{X(0) \times X(0)}$ as $D_S(x, x) = \vect_{S} (\type (x))$ for all  $ x \in X(0)$.  We prove that  the conditions of \cref{thm:matrix-trickle-down} hold  for $M_\tau \coloneqq  \frac{\Pi_\tau D_S}{k-1}$ for an arbitrary face $\tau \in X$ of co-dimension at least $k \geq 2$ and type $S$. If $G_S$ is connected, $\max_{\ti \in \prtin} \vect_S (\ti) \leq \frac{(k-1)^2}{3k -1}$ holds by assumption.  If $G_S$ is disconnected, $\max_{\ti \in \prtin} \vect_S (\ti) \leq \frac{(k-1)^2}{3k -1}$   follows from the assumptions that $\vect_S = \sum_{1 \leq i \leq \ell: |I_i| \geq 2}  \vect_{[d] \setminus I_i},$ where $I_1 \cup \dots \cup I_\ell $ are the vertices of connected components of $G_S$. That is because the supports of vectors $\vect_{[d] \setminus I_i}$ are disjoint by assumption and $\frac{(k-1)^2}{3k-1}$ is an increasing function for $k\geq 2$.  
        So, we get  $D_\tau \preceq \frac{(k-1)^2}{3k -1} I $, and thus,   $M_\tau \preceq \frac{k-1}{3k -1} \Pi_\tau$. To prove the rest of the conditions hold, first assume that $k =2$.  If $G_S$ is two disconnected vertices, we get  $f_S = 0$, and therefore, $D_S = 0$. Thus, we  get  $\Pi_\tau P_\tau - \pi_{\tau, 0}\pi_{\tau, 0}^\top \preceq 0 = \Pi_\tau D_S = M_\tau$, as desired.   If $G_S$ is connected, the base case assumption (\cref{item:cond1}) implies that $\lambda_2 (P_\tau) \leq  D_S (x, x)$ for all $x \in X_\tau (0)$.  Therefore, $\Pi_\tau P_\tau - \pi_{\tau, 0}\pi_{\tau, 0}^\top \preceq \Pi_\tau D_S = M_\tau $.   Now, assume $k \geq 3$. First assume that $G_S$ is disconnected and $\graph[I_1], \dots, \graph[I_\ell]$ are its connected components for some partition $I_1 \cup \dots \cup I_\ell$ of $\prtin \setminus S$.  Fix   any  $\sigma \in X_\tau (k-1)$. By \cref{lem:main_product},  $(X_\tau, \pi_{\tau}) =    (X_{\tau \cup \sigma_{-1}}, \pi_{\tau \cup  \sigma_{-1}}) \times \dots \times(X_{\tau \cup \sigma_{-\ell}},  \pi_{\tau \cup  \sigma_{-\ell}})$ where for every $1 \leq j \leq \ell $, $\sigma_{-j}$ is a subset of $\sigma$ that has type $ \prtin \setminus \left (S \cup I_j \right)$.  Therefore, we get   $\Pr_{\eta \sim \pi_{\tau \cup \sigma_{-j}}} [x\in \eta] = \Pr_{\eta \sim \pi_{\tau}} [x\in \eta]$ for all $1 \leq j\leq \ell$ and $x \in X_{\tau \cup \sigma_{-j}} (0)$. Combining this with \cref{fact:partite_normalize}, we get $k_j  \cdot  \pi_{\tau \cup \sigma_{-j}, 0} (x) = k \cdot \pi_{\tau, 0  } (x)$,  where $k_j = |I_j|$ for all $1 \leq j\leq \ell$. Thus we can write
        \begin{align*}
            \sum_{1 \leq j \leq \ell : |I_j| \geq 2} \frac{(k_j-1) k_j}{(k-1) k} M_{\tau  \cup \sigma_{-j}} &\underset{\text{def of } M_{\tau\cup_{\sigma_{-j}}}} =   \sum_{1 \leq j \leq \ell : |I_j| \geq 2} \frac{(k_j-1) k_j}{(k-1) k} \frac{\Pi_{\tau\cup_{\sigma_{-j}}} D_{\prtin \setminus I_j}}{k_j-1}  \\  &=  \sum_{1 \leq j \leq \ell : |I_j| \geq 2}  \frac{ k_j}{ k (k-1)}  \frac{k}{k_j} \Pi_\tau  D_{\prtin \setminus I_j} \\ &= \frac{ \Pi_\tau }{k-1}   \sum_{1 \leq j \leq \ell : |I_j| \geq 2}    D_{\prtin \setminus I_j} 
            = \frac{\Pi_\tau D_S}{k-1}=M_\tau, 
        \end{align*}
        where in the second to last equality, we used the fact that $\sum_{1 \leq j \leq \ell : |I_j| \geq 2}    f_{\prtin \setminus I_j} = f_S$, and thus $\sum_{1 \leq j \leq \ell : |I_j| \geq 2}    D_{\prtin \setminus I_j} = D_S$ by definition of $D_S$. 
        Now, assume that $G_S$ is connected. It is enough to show that   $\mathbb{E}_{x \sim \pi_{\tau, 0} }  M_{\tau \cup x}   \preceq   M_{\tau} -  M_\tau\Pi_\tau^{-1} M_\tau$.  This is equivalent to showing that for any $x \in X_\tau (0)$  
        \begin{align} \label{eq:Dcond}
            & \mathbb{E}_{y \sim \pi_{\tau, 0}} \left[  \frac{ (\Pi_{\tau}^{-1}\Pi_{\tau \cup y} D_{S \cup  \type (y) }) (x, x)}{k - 2 } \right]\leq    \frac{D_{S} (x, x) }{k - 1}  -  \frac{ D^2_{S} (x, x)} {(k - 2)(k - 1)}   
            \end{align}
        One can check that   for any $x \in X_\tau (0)$ of type $ \ti$
        \begin{align*}
        \mathbb{E}_{y \sim \pi_{\tau, 0}} \left[  \frac{ \Pi_{\tau}^{-1} \Pi_{\tau \cup y} D_{S \cup  \type (y) } (x, x)}{k - 2 } \right] &=  \frac{\sum_{y \in X_{\tau\cup x}(0)} \Pr_{\sigma \sim \pi_{\tau \cup x}} [y\in \sigma]  D_{S  \cup \type (y)}(x, x)}{(k-1)(k-2)} \\  
        &= \sum_{ \tj \in \prtin \setminus S} \frac{ \vect_{\tau \cup \tj}(i )} {(k-1)(k-2)}  \sum_{{\substack{y\in X_{\tau\cup x}(0):\\ \type(y) = \tj}}}\Pr_{\sigma \sim \pi_{\tau \cup x}} [y\in \sigma] 
        \\& = \frac{\sum_{ \tj \in \prtin \setminus S} \vect_{\tau \cup \tj}(i)}{(k-1)(k-2)},
        \end{align*}
        where in the last equality, we used \cref{fact:partite_sum_to_1}. Thus,  substituting  $D_{S} (x, x) = \vect_S (\type(x))$ in the RHS of \eqref{eq:Dcond}, it is enough to show that for any $\ti \in \prtin \setminus S$
        \begin{align*}
        \frac{\sum_{ \tj \in \prtin \setminus S} \vect_{\tau \cup \tj}(\ti)}{(k-1)(k-2)} \leq \frac{\vect_{S} (\ti)}{k - 1} -   \frac{\vect^2_{S}(\ti) }{(k -1)(k -2) },
        \end{align*}
        which holds by assumption \cref{item:cond2}. 
    \end{proof}

\section{Proof of Main Theorem}
    We are ready to prove \cref{thm:main_technical}. 
    
    \begin{proof} [Proof of \cref{thm:main_technical}]
         We find a family of vectors   $\{\vect_S \in \mathbb{R}^{\prtin}\}_{S \subset \prtin: |S| < \dimension}$ that satisfy the conditions of theorem \cref{prop:Dcondition}.
         Let $G \coloneqq G_{(X, \pi)}$. 
        Based on the  conditions of \cref{prop:Dcondition}, vectors 
        $\{\vect_S \in \mathbb{R}^{\prtin}\}_{S \subset \prtin: |S| < \dimension}$  can be defined as functions of  $\{\eps_{\{i,j\}}\}_{i, j\in \prtin, i\neq j}$. Recall that edges of $G$ capture pairs $\{i,j\}$ for which $\epsilon_{\{i,j\}} > 0$.  Assign every edge $\{i,j\}$ of $G$ with weight $\epsilon_{\{i,j\}}$. 
        We restrict our attention to functions that are very local with respect to $G$, i.e. for every $S$ and $i \in \prtin \setminus S$, we assume $f_S(i)$ only depends on $\Delta_S(i)$ and the weights of edges adjacent to $i$ in $G_S$  if $\Delta(i) >1$. It turns out that if $\Delta(i) =1$, we would need to also take into account the degree of the unique  neighbor of $i$.   More formally, consider the following family of vectors $\{\vect_S \in \mathbb{R}^{\prtin}\}_{S \subset \prtin: |S| < \dimension}$: 
        for any $S \subset \prtin$ such that  $|S| < \dimension$, let $\vect_S$  be of the following form: for any $\ti \in S$, let $\vect_S(\ti) = 0$,  and for any $\ti \in  \prtin \setminus S$ let 
        \begin{align*}
            \vect_S(\ti) \coloneqq
        \begin{cases*}
            0 & if $\Delta_S(\ti) =0$, \\
            \epsilon_{\{\ti, \tj\}} \cdot g_{\ti, \tj}(\Delta_S(\tj)) & if $\Delta_S(\ti) =1 $ and $\ti \sim_S \tj$, \\
          \sum_{\tj \sim_S \ti} \epsilon_{\{\ti, \tj\}} \cdot  h_{\ti} \left(\Delta_S (\ti) \right)   & if $\Delta_S(\ti) \geq 2$, 
        \end{cases*}
        \end{align*}
        where for every  $\ti \in \prtin$ and $\tj \sim \ti$, functions $g_{\ti, \tj}, h_{\ti}: \{1, \dots, \Delta\}  \rightarrow \mathbb{R}_{\geq 0}$  are defined later in a way that  guarantees that   $\{\vect_S\}_{S \subset \prtin: |S| < \dimension}$ satisfies the assumptions of \cref{prop:Dcondition} (see \eqref{eq:fdef}, \eqref{eq:hdef}). 
        \par
        First, consider the case that  $G_S$ is disconnected.  Note that  for any $S, S' \subset \prtin$ such that  $|S|, |S'| < \dimension$, if $\{\tj \in \prtin: \tj \sim_{S} \ti\}= \{\tj \in \prtin: \tj \sim_{S'} \ti\}$ for some $\ti \notin S, S'$, then $\vect_S(\ti) = \vect_{S'}(\ti) $. Let $I_1, \dots, I_\ell$ be the vertices of connected components of $G_S$.  Since the neighborhood of each vertex in any connected component of $G_S$ is the same as its neighborhood in $G_S$, we get $\vect_S = \sum_{1 \leq i \leq \ell: |I_i| \geq 2}  \vect_{[d] \setminus I_i}$. 
        \par
          Now, assume $G_S$ is connected. 
          Take an arbitrary $k \geq 2$ and  $S \subset \prtin$ of size $(\dimension+1) - k $. First we  verify  the set of conditions given in \cref{item:cond1} and \cref{item:cond2}.
          First, assume that $k=2$.  Let  $ \prtin \setminus S = \{\ti, \tj\}$. By definition of $\epsilon_{\{\ti, \tj\}}$,   for any $\tau$ of type $S$, $\lambda_2(P_\tau) \leq  \epsilon_{\{\ti, \tj\}}$. Thus, if we define $g_{\tl, \ttt} (1) = 1$ for all distinct $\tl, \ttt \in \prtin$, then we get   $\lambda_2(P_\tau)\leq  \epsilon_{\{\ti, \tj\}} = \epsilon_{\{\ti, \tj\}} g_{\ti, \tj} (1)  =  \vect_S(\ti) =  \vect_S(\tj)$, as desired. Now, assume that $k\geq 3$. Fix  an arbitrary $\ti \in \prtin \setminus S$. Our goal is to define  $g_{\ti, \tj}, h_\ti: \{1, \dots, \Delta\}  \rightarrow \mathbb{R}_{\geq 0}$ for all $\tj \sim \ti$ such that $g_{\ti, \tj} (1) = 1$ for all  $\tj \sim \ti$ and  the following inequality is satisfied: 
                \begin{align}\label{eq:main-rec}
                    \sum_{\tj \in \prtin \setminus (S \cup \ti)}   \vect_{S \cup \tj} (\ti) \leq  (k-2) \vect_S (\ti)- \vect_S^2 (\ti).
                \end{align}
        To keep the notation concise, relabel the elements such that $i$ is relabeled to $0$  and  $\epsilon_{\{0, 1\}} \geq \dots \geq \epsilon_{\{0, \dimension\}}$. Moreover, let $\epsilon_\tj \coloneqq \epsilon_{\{0, \tj\}}$ for any $\tj \in \prtin \setminus 0$. 
        \par
        {\bf Case 1:  $\Delta_S (0) = 1$, and $ \tj \sim_S 0$.}  Since $G_S$ is connected and $(\dimension+1) - |S| \geq 3$, we have  $\Delta_S( \tj) \geq 2$. Let $t \coloneqq \Delta_S (\tj)$. We have
        \begin{align*}
            \sum_{\tl \in \prtin \setminus (S \cup 0)}  \vect_{S \cup \tl} (0) &= \vect_{S \cup \tj} (0)  + \sum_{\tl \in \prtin \setminus (S \cup 0): \tl \sim_S \tj}  \vect_{S \cup \tl} (0) +  \sum_{\tl \in \prtin \setminus (S \cup 0):  \tl \not\sim_S \tj, \tl \neq \tj}  \vect_{S \cup \tl} (0) \\& 
             = 0 + (t-1) \cdot  \epsilon_{\tj} \cdot g_{{0, \tj}}(t -1) + (k - t -1) \cdot  \epsilon_{\tj} \cdot g_{{0, \tj}}(t ).
        \end{align*}
        On the other hand,  $(k-2) \vect_S(0) - \vect_S(0)^2 =  (k-2) \cdot \epsilon_{\tj} \cdot g_{0, \tj} (t) - \epsilon^2_{\tj} \cdot g^2_{0, \tj}(t)$.    
        So  it is enough to satisfy 
        \begin{align}\label{eq:case1finaleq}
          (t-1) \cdot  \epsilon_{ \tj} \cdot ( g_{0, \tj} (t) - g_{0, \tj} (t -1) )\geq  \epsilon^2_{\tj } \cdot  g_{0, \tj}^2(t). 
        \end{align}
        Now, define $g_{0, \tj}: \{1, \dots, \Delta\} \rightarrow \mathbb{R}_{\geq 0}$ as follows:  recall that we defined $g_{0, \tj}  (1) =1$. For any $2 \leq \ell \leq \Delta$, let 
        \begin{align} \label{eq:fdef}
            g_{0, \tj} (\ell) = 1 + 1.3 \cdot \epsilon_{\tj}  \cdot H_{\ell-1} . 
        \end{align}
        Using assumption \eqref{eq:thmcond1}, $\epsilon_{\tj} H_{\Delta -1}   \leq \frac{\delta^2}{10} \leq \frac{1}{10}$. Thus
        \begin{align*} 
            \epsilon^2_\tj \cdot  g^2_{0, \tj} (t) \leq  \epsilon_{\tj}^2 \left(1 + 1.3 \epsilon_\tj\left(1+  H_{\Delta -1} \right)\right)^2  < 1.3 \epsilon_\tj^2.
        \end{align*}
        Substituting $ g_{0, \tj}(t)$ according to \eqref{eq:fdef} and using the above bound, one can verify that   \eqref{eq:case1finaleq} holds. 
        \par
        {\bf Case 2:  $\Delta_S (0) \geq 2$.}
         For simplicity of notation, let $t \coloneqq \Delta_S (0)$ and $\alpha \coloneqq \sum_{ \tj: \tj \sim_S 0 } \epsilon_\tj$.  Define $h_0 (1) =\max_{\tj: \tj \sim 0} g_{0,\tj}(\Delta)$. 
        \begin{align*}
            \sum_{\tj \in \prtin \setminus (S \cup 0)}  \vect_{S \cup \tj} (0) & =  \sum_{\tj \in \prtin \setminus (S \cup 0): \tj \sim_S 0}  \vect_{S \cup \tj} (0) +  \sum_{\tj \in \prtin \setminus (S \cup 0):  \tj \not\sim_S 0}  \vect_{S \cup \tj} (0) 
            \\&\leq    \left(\sum_{\tj \in \prtin \setminus (S \cup 0): \tj \sim_S 0} (\alpha - \epsilon_{\{0, \tj\}}) \right)  \cdot h_0 (t-1) +  (k- t-1) \cdot \alpha  \cdot h_0 (t) \\& =  (t-1) \cdot \alpha  \cdot h_0 (t-1) +  (k- t-1) \cdot \alpha \cdot  h_0 (t). 
        \end{align*}
        Note that if $t\geq 3$, the first inequality is an equality by definition. If $t = 2$, the  first inequality follows from the definition of $h_0(1)$. 
        Thus, it is enough to satisfy 
         \begin{align*}
             \sum_{\tj \in \prtin \setminus (S \cup 0)}  \vect_{S \cup \tj} (0) &= (t-1) \cdot \alpha  \cdot h_0 (t-1) +  (k- t-1) \cdot \alpha \cdot  h_0 (t) \\& \leq    (k-2) \cdot \alpha  \cdot h_0 (t ) - \alpha^2  \cdot h_0^2 (t ) = (k-2)  \vect_{S} (0)-   \vect_{S}^2 (0). 
        \end{align*}
        Equivalently, it suffices to satisfy  
        \begin{align}\label{eq:case2finaleq}
             (t-1) (h_0 (t) - h_0 (t-1)) \geq   \alpha  \cdot h_0^2 (t ). 
        \end{align}
        Now, define $h_0 : \{1, \dots, \Delta\} \rightarrow \mathbb{R}_{\geq 0}$ as follows: recall that we defined $h_0 (1) =\max_{\tj: \tj \sim 0} g_{0,\tj}(\Delta)$. For any $2 \leq \ell \leq \Delta$, let
        \begin{align}\label{eq:hdef}
            h_0(\ell) \coloneqq \frac{h_0(1)}{1 - \const \left( 
            \sum_{\tj=1}^\ell \epsilon_\tj H_{\tl-1} (\tj-1) \right)}.
        \end{align}
        We need to prove \eqref{eq:case2finaleq} for a carefully chosen $\const$.  Let $\beta$ be such that 
         $h_0 (t)  = \frac{h_0(1)}{\beta}$. We get  $h_0 (t-1) = \frac{h_0(1)}{\beta+\const (\sum_{\tj =1}^t \frac{\epsilon_\tj}{t-1}) }$, and thus,  
        \begin{align*}
            (t-1) (h_0 (t) - h_0 (t-1)) = \frac{ h_0(1) \cdot \const \sum_{\tj=1}^t \epsilon_\tj }{\beta \cdot (\beta + \frac{\const \sum_{\tj=1}^t \epsilon_\tj}{t-1} )}.   
        \end{align*}
        Note that $\alpha \cdot h^2_0 (t) = \frac{\alpha \cdot  h^2_0(1) }{\beta^2}$. Thus,  to satisfy \eqref{eq:case2finaleq}, it is enough to show that 
        \begin{align*}
           \beta  \cdot  \const  \cdot \left(\sum_{\tj=1}^t \epsilon_\tj\right) \geq  \alpha \cdot  h_0(1)  \cdot \left(\beta + \frac{\const \sum_{\tj=1}^t \epsilon_\tj}{t-1}\right). 
        \end{align*}
        Note that 
        \begin{align} \label{eq:boundh1}
             h_0(1) \leq \max_{\tj \sim \ti} g_{0, \tj} (\Delta) =   1+ 1.3 \epsilon_1 H_{\Delta-1} \underset{\text{by } \eqref{eq:thmcond1}}{\leq}  1 + 1.3 \frac{\delta^2}{10}.
        \end{align}
Moreover, $ \sum_{\tj=1}^t \epsilon_\tj \geq  \sum_{\tj: \tj \sim_S 0} \epsilon_\tj = \alpha$. Thus, letting $\const=1+c'\delta$ for some $c'>0$ that we choose later, it is enough to show that 
         \begin{align*}
           \beta  \cdot  (c'-0.13 \delta)\delta  \geq   (1+0.13\delta) \cdot  \frac{(1+c'\delta) \sum_{\tj=1}^t \epsilon_\tj}{t-1}. 
        \end{align*}
        Using $\frac{\sum_{\tj=1}^t \epsilon_\tj}{t-1} \leq 2 \epsilon_1 
        \underset{\eqref{eq:thmcond1}}{\leq} \frac{\delta^2}{5}$, it is enough to show that
         \begin{align}\label{eq:main-condition-delta}
           \beta  \cdot  (c'-0.13 \delta)  \geq   (1+0.13\delta)(1+c'\delta)  \frac{\delta}{5}.
           \end{align}
        On the other hand,
         \begin{align}\label{eq:beta} 
            \beta &\geq 1 - (1+c'\delta)  \left( 
              \sum_{\tj=1}^{\Delta(0)} \epsilon_\tj H_{\Delta(0)-1} (\tj-1)  \right)  \underset{\eqref{eq:thmcond2}}{\geq} 1- (1+c'\delta)(1-\delta ) =   \delta  (1- c' + c'\delta),
        \end{align}
        Thus, to satisfy \eqref{eq:main-condition-delta}, it is enough to show that  $  (1- c'+ c' \delta) (c' -0.13 \delta)\geq (1+1.13 \delta) (1+c'\delta)\frac{1}{5}$. Letting $c'  = \frac{1}{2}$,  this inequality holds  for every $0 < \delta < 1$. 
This establishes $\cref{eq:main-rec}$. So we verified conditions \cref{item:cond1} and \cref{item:cond2} are satisfied. 
\par

To show that all conditions of \cref{prop:Dcondition} are satisfied, it remains to show that  $\max_{\ti \in \prtin} \vect_S (\ti) \leq \frac{(k-1)^2}{3k -1}$.   
        Note that $\sum_{\tj: \tj \sim \ti} \epsilon_{\{\ti,\tj\}}  \leq  \Delta_S \cdot \epsilon_1  \underset{\eqref{eq:thmcond1}}{\leq} \Delta_S \cdot \frac{\delta^2}{10}$ for all $\ti \in \prtin \setminus S$. 
        Thus, we get $\max_{\ti \in \prtin} \vect_S  (\ti) \leq \Delta_S \cdot \frac{\delta^2}{10}  \max_{\ti \in \prtin \setminus S} \cdot h_\ti (\Delta_S(\ti))$. Moreover, using  \eqref{eq:boundh1} and  \eqref{eq:beta} with $c' = \frac{1}{2}$ (we can write this inequality for every $\ti$), we get  
        \begin{align}\label{eq:boundh}
              h_\ti (\Delta_S(\ti)) \leq h_\ti (\Delta(\ti)) \leq  \frac{1+\frac{\delta^2}{10}}{\delta(\frac{1}{2}+ \frac{\delta}{2})},
        \end{align}
       Thus, we can write 
        \begin{align*}
            \max_{\ti \in \prtin} \vect_S  (\ti) &\leq \Delta_S \cdot \frac{\delta^2}{10}  \frac{ 1+\frac{\delta^2}{10} }{\delta (\frac{1}{2} + \frac{\delta}{2})} \leq  \frac{\Delta_S}{5} \leq \frac{k-1}{5}\leq \frac{(k-1)^2}{3k -1},
        \end{align*}
        as desired. 
        So we proved that  $\{\vect_S\}_{S \subset \prtin: |S| < \dimension}$ satisfies the conditions of \cref{prop:Dcondition}. Now, we are ready to  bound   $\lambda_2(P_\tau)$ for any face  $\tau$ of co-dimension $k\geq 2$ and type $S$. 
        First, we show  that for every $\ti \in \prtin \setminus S$, $ \sum_{\tj : \tj \sim_S \ti} \epsilon_{\{\ti, \tj\}}  \leq 1-\delta$. Note that  $$\sum_{\ell=1}^{\Delta(i)} H_{\Delta(\ti)-1}(\ell-1) = \sum_{\ell =2}^{\Delta (\ti) } \frac{\ell}{\ell-1} = 2 + \sum_{\ell =3}^{\Delta (\ti) } \frac{\ell}{\ell-1} \geq \Delta (\ti).$$ Thus, we can write 
        \begin{align}\label{eq:sumepsij}
        \sum_{\tj: \tj \sim_S \ti} \epsilon_{\{\ti, \tj\}} \leq \left(\sum_{\ell=1}^{\Delta(\ti)}\frac{H_{\Delta(\ti)-1}(\ell-1)} {\Delta(\ti)}  \right) \left( \sum_{\tj \sim \ti} \epsilon_{\{\ti, \tj\}}   \right)\leq  \sum_{\ell=1}^{\Delta (\ti)} H_{\Delta(\ti) -1}(\ell-1)  \cdot \epsilon_{\{\ti,\tj_\ell\}} \leq 1-\delta. 
        \end{align}
        where we assumed that  $\ti_1, \dots, \tj_d$ is an ordering of $\prtin \setminus S$ such that $\epsilon_{\tj_1}\leq \dots \leq  \epsilon_{\tj_\dimension}$. Using this inequality and \eqref{eq:boundh}, we get 

             \begin{align*}
            \lambda_2(P_\tau) &\leq \frac{ \max_{\ti \in \prtin \setminus S} \vect_S(\ti) }{k-1} \leq  \frac{ \max_{\ti \in \prtin} ( \sum_{\tj \sim_S \ti} \epsilon_{\{\ti, \tj\}}) \cdot h_\ti (\Delta_S(\ti)) }{k-1} \\& \leq \frac{ (1-\delta)\cdot \max_{\ti \in \prtin} h_\ti (\Delta(\ti) ) }{k-1} \leq   \frac{(1-\delta) \cdot \frac{2(1+\frac{\delta^2}{10}\delta)}{\delta (\delta+1)}}{k-1},
        \end{align*}
        as desired. 
    \end{proof}

\section{Obstructions to Trickle Down} 
 In this section, we exhibit some barriers to any class of trickle down theorems that only look at 2nd eigenvalue of links of co-dimension 2. 
 
        \begin{example} \label{ex:hardcore}
We exhibit a family of partite complexes such that $\gamma_2\leq \lambda$, but $\gamma_{2d+1}\geq \Omega(\lambda)$.  Consider a $2d$-partite complex $(X,\pi)$ on the ground set  $\cup_{i}^{2d} T_i$, where for all $i \in \{1, \dots, d\}$, $T_i = \{i_{in}, i_{out}\}$ are the elements of type $i$. A set $\tau\in \prt_1\times \dots\times \prt_{2d}$ is a facet of $X$ iff  for all $i_{in},j_{in}\in \tau$ we either have $i_{in},j_{in}\leq d$ or $i_{in},j_{in}\geq d+1$. 
For such a $\tau$, we define  $w(\tau)=\lambda^{\norm{\tau}_1}$ where $\norm{\tau}_1$ is the number of $i\in \tau$, for some $0<\lambda<1$ that we choose later (the weight of the complex is the probability distribution induced by this weight function, but for simplicity of calculations,  we do not normalize the weights here).

As a side note, facets of this complex corresponds to the set of independent sets of $K_{d,d}$ the complete bipartite graph on the sets $\{1,\dots,d\}, \{d+1,\dots,2d\}$. It is not hard to see that we have $\gamma_2=\frac{\lambda}{1+\lambda}$. 
To see this, note that a  worst link  to take is  $\tau=\{2_{out},\dots,(2d-1)_{out}\}$ with the the following 1-skeleton. 

Notice that this graph  is a $\frac{\lambda}{1+\lambda}$-spectral expander.
\begin{figure}[htb]\centering
\begin{tikzpicture}
	\node [draw,circle] at (0,0) (a) {\small $1$};
	\node [draw,circle] at (2,0) (b) {\footnotesize $\overline{2d}$} edge node [above] {\small $\lambda$}(a);
	\node [draw,circle] at (4,0) (c)  {\small $\bar{1}$} edge node [above] {\small $1$}(b); 
	\node [draw,circle] at (6,0) (d) {\footnotesize $2d$} edge node [above] {\small $\lambda$}(c);
\end{tikzpicture}	
\end{figure}

We claim that the 2nd eigenvalue of the 1-skeleton of the link of the empty set is at least $\Omega(\frac{\lambda}{1+\lambda})$. 
First notice that this graph has $4d$ vertices. We partition its vertices into 4 sets, $A=\{1_{in},\dots,d_{in}\}, \bar{A}=\{1_{out},\dots,d_{out}\}$ and $B=\{(d+1)_{in},\dots,(2d)_{in}\}, \bar{B}=\{(d+1)_{out},\dots,(2d)_{out}\}$.

For simplicity of calculations, we can write the weight of every edge of the 1-skeleton as the sum of the weights of all facets that contain that edge.
For every $i_{in} \in A, j_{out}\in \bar{A}$,
$$ 
w_{i_{in},j_{out}} = 
\sum_{k=0}^{d-2} {d-2 \choose k} \lambda^{k+1}=\lambda (1+\lambda)^{d-2}$$
Running a similar calculation for all possible  pairs of elements, we obtain the following 1-skeleton of the empty set (after dividing all edge weights by $(1+\lambda)^{d-2}$).
\begin{figure}[htb]\centering
\begin{tikzpicture}
	\node [draw,circle] at (0,0) (a) {$\overline{A}$};
	\node [draw,circle] at (3,0) (b) {$\overline{B}$} edge [color=red,line width=1.4pt] node [above=0.3cm] {\small $2(1+\lambda)-(1+\lambda)^{-(d-2)}$}(a);
	\node [draw,circle] at (0,-3) (c) {$A$} ;
	\node [draw,circle] at (3,-3) (d) {$B$};
	\path [line width=1.2pt,color=blue] (c) edge  node [left] {\small$\lambda$}(a) 	(d) edge (b);
	\path [line width=1.2pt, color=black] 
	(c) edge node [below=0.7cm] {\small $\lambda(1+\lambda)$}(b)
	(d) edge (a);
	\path [color=orange, line width=1.1pt] (c) edge[loop left] node [left] {\small $\lambda^2$} (c) 
	(d) edge [loop right] (d); 
	\path [color=purple, line width=1.1pt] (a) edge [loop left] node [left] {\small $1+(1+\lambda)^2- (1+\lambda)^{-(d-2)}$}(a)
	(b) edge [loop right] (b);
\end{tikzpicture}
\caption{1-skeleton of the Link of $\emptyset$ for complex \cref{ex:hardcore}. Edges of the same color have the same weight.}
\end{figure}
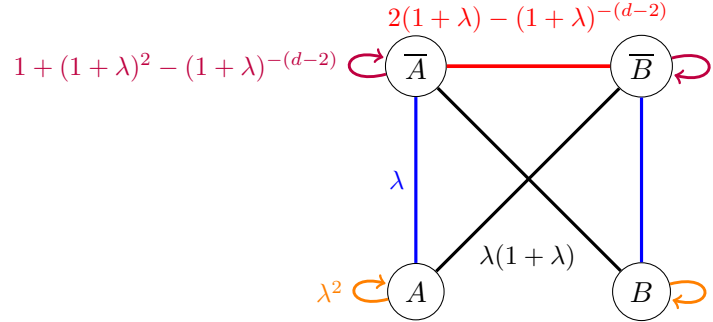

We assume $1/d \ll \lambda\ll 1$ and $d\to\infty$ so we ignore $(1+\lambda)^{-(d-2)}$ low order term.
It follows that $d_w(i_{in})\approx 2d\lambda(1+\lambda)$ and $d_w(i_{out})\approx 2d(2+3\lambda+\lambda^2)$ for $i$. It follows that,
$$ \left|w(E(A\cup B))-\frac{\vol(A\cup B)^2}{\vol(V)}\right|\geq  \left|d^2\lambda^2 - \frac{(2d(2d\lambda(1+\lambda)  ))^2}{2d(2d(2+4\lambda+2\lambda^2))}\right| = \left|d^2\lambda^2 - 2d^2\lambda^2 \right|$$
So, by \cref{fact:expmixinglemma} for $S=A\cup B$ we have
$$ \lambda_2 \geq \frac{d^2\lambda^2}{\vol(A\cup B)} = \frac{d^2\lambda^2}{4d^2\lambda(1+\lambda)} \geq \frac{\lambda}{4(1+\lambda)}.$$
\end{example}
Note that if we apply \cref{thm:main_technical} to the setting of the above example, we obtain $\Delta(i)=d$ for all $1\leq i\leq 2d$. So, if $\lambda\cdot d\leq 1-\delta$, (and $\lambda d\log d\leq \delta^2/10$), our theorem implies that the 2nd eigenvalue of the link of the emptyset is at most $\frac{c(1-\delta)}{2d\cdot \delta}\leq \frac{c\lambda}{\delta}$ which is consistent with the above calculations.

\begin{example}
In this example we construct a totally connected (non-partite) $(d-1)$-dimensional weighted simplicial complex $(X,\pi)$ such that links of co-dimension 2 are $0.5$-spectral expanders, but the 1-skeleton of the link of  the emptyset is only a $1-\text{exp}(-d)$-spectral expander.  
See \cite{LMSY23} for a random family of 3-dim complexes exhibiting a similar growth of eigenvalues.
This in particular shows that local spectral expansion can increase significantly for (non-partite) complexes. Let $B(V,E)$ be the  following {\em barbell} graph: Consider two disjoint cliques $K_1,K_2$ each with $2d$ vertices and connect them with a path of length $d+2$ namely $x_0\in K_1, x_1, \dots, x_{d},x_{d+1}\in K_2$. Note that $x_1,\dots,x_{d}$ do not belong to the cliques.
\begin{figure}[htb]\centering
\begin{tikzpicture}
	\draw (0,0) circle (1.5); 
	\draw (8,0) circle (1.5);
	\node [draw,circle,inner sep=2] at (1.2,0) (a) {} node at (1.2,-0.3) {\small $x_0$}; 
	\node [draw,circle,inner sep=2] at (2.2,0) (b) {} edge (a) node at (2.3,-0.3) {\small $x_1$}; 
	\node [draw,circle,inner sep=2] at (6.9,0) (c) {}  node at (6.7,-0.3) {\footnotesize $x_{d+1}$}; 
	\node [draw,circle,inner sep=2] at (5.5,0) (d) {} edge (c) node at (5.4,-0.3) {\footnotesize $x_{d}$};
	\node at (3,0) () {} edge (b);
	\node at (4.6,0) () {} edge (d);
	\draw [dotted, line width=1.4 pt](3.2,0)-- (4.4,0);
	\node at (-0.5,-.8) () {$K_1$};
	\node at (8.5,-.8) () {$K_2$};
\end{tikzpicture}	
\end{figure}

Now we define a $d-1$-dimensional weighted complex on vertices $V$. The facets of $X$ are precisely  sets $S\subseteq V$ with $|S|=d$ such that the induced graph $B[S]$ is connected. For simplicity of calculations we define the weight of every facet to be 1, i.e., we don't normalize the weights to be a probability distribution. 

We claim that any link $X_\tau$ of co-dimension 2 is $\Omega(1)$-spectral expander. Indeed the worst link is $X_\tau$ for $\tau=\{x_2,\dots,x_{d-1}\}$. The 1-skeleton of $X_\tau$ is the following graph with second eigenvalue $0.5$.
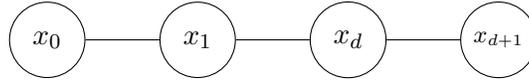
\begin{figure}[htb]\centering
\begin{tikzpicture}
	\node [inner sep=2mm,draw,circle] at (0,0) (a) {$x_0$};
	\node [inner sep=2mm,draw,circle]at (2,0) (b) {$x_1$} edge (a);
	\node [inner sep=2mm,draw,circle] at (4,0) (c) {$x_d$} edge (b);
	\node [inner sep=1.2mm,draw,circle] at (6,0) (d) {\footnotesize $x_{d+1}$} edge (c);
\end{tikzpicture}	
\end{figure}

Now, we claim that the 1-skeleton of $X_{\emptyset}$, $G_0$, has 2nd eigenvalue $1-\frac{1}{\text{exp}(d)}$. First notice $G_0$ has the same vertex set as $B$. 
Now, consider the set $S=K_1\cup \{x_0,x_1,\dots,x_{d/2}\}$.
For simplicity of calculations, we let the weight of every edge $\{x,y\}$ of $G_0$ be sum of the weights of all facets that contains $x,y$.
First, observe that the weight of every edge in $K_1$ is (at least) ${2d-2\choose d-2}$. 
 On the other hand, the weight of any edge in the cut $(S,\overline{S})$ is at most $2{2d\choose d/2}$.
Putting these together it follows that
$$ \phi(S)=\frac{w(E(S,\overline{S}))}{\vol(S)} \geq \frac{2\frac{d}{2}\cdot \frac{5d}{2}\cdot {2d\choose d/4}}{{2d\choose 2} \cdot {2d-2\choose d-2}}\approx 2^{-d/2},$$
for a large enough $d$. 
Therefore, by Cheeger's inequality, \cref{lem:cheeger}, the second eigenvalue of $G_0$ is at least $1-2^{-\Omega(d)}$ for $d$ large enough.
\end{example}

\printbibliography
\end{document}